\documentclass{article}

\usepackage[light,easyscsl,sfmathbb,nowarning]{kpfonts}
\usepackage{fullpage}
\usepackage{url}
\usepackage[hidelinks]{hyperref}
\usepackage[utf8]{inputenc}
\usepackage[small]{caption}
\usepackage{graphicx}
\usepackage{amsmath}
\usepackage{amsthm}
\usepackage{amsfonts}
\usepackage{booktabs}
\usepackage{amssymb,mathtools}
\usepackage{bm}
\usepackage[pdftex,dvipsnames]{xcolor}
\usepackage{authblk}
\usepackage{xspace}
\usepackage{comment}
\usepackage{array}
\usepackage{multirow}
\usepackage{nicefrac}
\usepackage{todonotes}
\usepackage[noend]{algpseudocode}
\usepackage[numbers,sort]{natbib}

\usepackage[british]{babel}
\usepackage{hyperref}
\hypersetup{
    colorlinks=true,
    linkcolor=olive,
    citecolor=purple
}

\hypersetup{pdfinfo={
  Title={Near-Tight Algorithms for the Chamberlin-Courant and Thiele Voting Rules},
  Author={Krzysztof Sornat, Virginia Vassilevska Williams and Yinzhan Xu},
  Keywords={multiwinner voting, single-peakedness, algorithms, lower bounds, Chamberlin-Courant, Thiele}
}}

\title{Near-Tight Algorithms for\\
  the Chamberlin-Courant and Thiele Voting Rules\thanks{A conference version of this work appears in IJCAI 2022~\cite{SornatVWX22}.}}
\date{}

\author[1]{Krzysztof Sornat\thanks{sornat@agh.edu.pl. Part of this work was done while Krzysztof was a postdoc at MIT CSAIL, USA.}}
\author[2]{Virginia {Vassilevska Williams}\thanks{virgi@mit.edu}}
\author[2]{Yinzhan Xu\thanks{xyzhan@mit.edu}}
\affil[1]{AGH University, Poland}
\affil[2]{MIT CSAIL, USA}

\newtheorem{theorem}{Theorem}
\newtheorem{claim}[theorem]{Claim}
\newtheorem{definition}[theorem]{Definition}

\newtheorem{lemma}[theorem]{Lemma}
\newtheorem{corollary}[theorem]{Corollary}
\newtheorem{observation}[theorem]{Observation}

\newcommand{\ProblemFormat}[1]{{\sc #1}}
\newcommand{\ProblemName}[1]{\ProblemFormat{#1}\xspace}
\newcommand{\chamberlincourant}[0]{\ProblemName{Chamberlin-Courant}}
\newcommand{\thiele}[0]{\ProblemName{Thiele}}
\newcommand{\generalizedthiele}[0]{\ProblemName{Generalized Thiele}}
\newcommand{\kmedian}[0]{\ProblemName{$k$-Median}}
\newcommand{\maxkcoverage}[0]{\ProblemName{Max $k$-Coverage}}
\newcommand{\cnfsat}[0]{\ProblemName{CNF-SAT}}
\newcommand{\tsat}[0]{\ProblemName{$3$-SAT}}
\newcommand{\independentset}[0]{\ProblemName{Independent Set}}
\newcommand{\hittingset}[0]{\ProblemName{Hitting Set}}

\newcommand{\Oh}{\ensuremath{\mathcal{O}}}
\newcommand{\Ohstar}{\ensuremath{\mathcal{O}^*}}
\newcommand{\tO}{\ensuremath{\tilde{\Oh}}}

\newcommand{\calF}{\ensuremath{\mathcal{F}}}

\newcommand{\ComplexityFont}[1]{{\ensuremath{\mathsf{#1}}}}

\newcommand{\classW}{\ComplexityFont{W}}
\newcommand{\classNP}{\ComplexityFont{NP}}

\newcommand{\naturals}{{{\mathbb{N}}}}
\newcommand{\reals}{{{\mathbb{R}}}}

\DeclareMathOperator*{\sort}{Int\-Sort}

\newcommand{\eps}{\varepsilon}

\DeclareMathOperator*{\poly}{poly}

\begin{document}

\maketitle

\begin{abstract}
We present an almost optimal algorithm for the classic Chamberlin-Courant multiwinner voting rule (CC) on single-peaked preference profiles.
Given $n$ voters and $m$ candidates, it runs in almost linear time in the input size, improving the previous best $\Oh(nm^2)$ time algorithm of Betzler et al.~(2013).
We also study multiwinner voting rules on nearly single-peaked preference profiles in terms of the candidate-deletion operation.
We show a polynomial-time algorithm for CC where a given candidate-deletion set $D$ has logarithmic size.
Actually, our algorithm runs in $2^{|D|} \cdot \poly(n,m)$ time and the base of the power cannot be improved under the Strong Exponential Time Hypothesis.
We also adapt these results to all non-constant Thiele rules which generalize CC with approval ballots.
\end{abstract}

\section{Introduction}
We study computational aspects of the Chamberlin-Courant voting rule (CC)~\cite{ChamberlinCourant83} with general misrepresentation function and the Thiele rules~\cite{Thiele95} in which we choose a committee of size $k$ such that the total utility of the voters is maximized.
In the case of CC, the utility of a voter from a committee is determined by the most preferred committee member.
In the case of a $w$-Thiele rule, parameterized by a non-increasing sequence $w=(w_1,w_2,\dots)$, the utility of a voter is equal to $w_1 + w_2 + \dots + w_x$, where $x$ is the number of committee members approved by the voter.

CC and Thiele rules are basic multiwinner voting rules studied in social choice theory and, in particular, in the computational social choice community~\cite{FaliszewskiSST17trends,BrillLS17}.
They were discussed broadly for their applications, not only in parliamentary elections but also in many scenarios when a group of agents has to pick some number of items for joint use~\cite{SkowronFL16}.

Unfortunately, computing an optimal committee under CC is already \classNP-hard for the case of {\it approval ballots}, i.e., when each voter gives a subset of approved candidates~\cite{ProcacciaRZ08}.
This special case is called {\it Approval Chamberlin-Courant} (Approval-CC) and it was considered by Procaccia et al.~\cite{ProcacciaRZ08}
as a minimization problem.
Its maximization version is equivalent to the well-known \maxkcoverage problem for which many hardness results have been shown \cite{SkowronF17}.
In particular, \maxkcoverage is \classNP-hard\footnote{\classNP-hardness of \maxkcoverage directly implies \classNP-hardness of Approval-CC.
This consequence was mentioned (without a proof) in the conference paper of Procaccia et al.~\cite{ProcacciaRZ07}.}
to approximate within $(1-1/e+\eps)$ factor for any $\eps>0$~\cite{Feige98}.
Moreover, this hardness of approximation holds even in FPT-time with respect to solution size $k$ assuming a gap version of the Exponential Time Hypothesis~\cite{Manurangsi20}.

Notice that Approval-CC is equivalent to the $(1,0,\dots)$-Thiele rule, so all these hardness results for Approval-CC hold for this basic Thiele rule as well.
One may consider other $w$-Thiele rules with specific $w$, e.g., $w$-Thiele rule with $w_i = 1/i$ which is equivalent to the well-known {\it Proportional Approval Voting}~\cite{Kilgour2010} and $w_i = 1/i^2$ which is related to {\it Penrose apportionment method}~\cite{BrillLS17}.
Notice that $(1,1,\dots)$-Thiele is equivalent to {\it Multiwinner Approval Voting}~\cite{AzizGGMMW15}, which is solvable in polynomial-time.
All the other ones\footnote{
We assume $w_1=1$ without loss of generality.
} are \classNP-hard~\cite[Theorem 5]{SkowronFL16}.

Analogously to \maxkcoverage, hardness of approximation results hold for the natural class of $w$ sequences such that partial sums of $w$ grows in a sublinear way,
where a hardness constant depends strictly on $w$~\cite{DudyczMMS20,BarmanFF21,BarmanFGG22}.
It follows that for such Thiele rules one should not expect a PTAS nor even an FPT-approximation scheme.

Because of the hardness of CC and Thiele rules, both rules were studied in restricted domains, i.e., on preference profiles with certain structures, e.g., single-peaked (SP) and single-crossing (SC).
For a survey on structured preferences, see~\cite{ElkindLP17trends}.
Such restricted domains are motivated by real-world elections.
For instance, SP preference profiles appear when, intuitively, the candidates can be placed on an one-dimensional axis, and the farther away a candidate is from a voter's favorite candidate the more the voter dislikes the candidate.
Examples of such SP-axis are different policies in an ideological spectrum: the Left vs the Right, more liberal vs more conservative etc.

Fortunately, CC and Thiele rules on SP domain (CC-SP and Thiele-SP respectively) are known to be polynomial-time solvable~\cite{BetzlerSU13,Peters18}.
However, the dynamic programming algorithm for CC-SP proposed by
Betzler et al.~\cite{BetzlerSU13} is not optimal in terms of running time.
We fill this gap by showing an almost optimal algorithm and we study polynomial-time solvability for nearly single-peaked preference profiles for CC and Thiele rules.

\subsection{Our Contribution}
The previous best running time for CC-SP was $\Oh(nm^2)$~\cite{BetzlerSU13}.
We improve it by presenting an almost linear time algorithm (up to subpolynomial factors) in Theorem~\ref{thm:cc-sp}.
To achieve this result we reduce our problem to the Minimum Weight $t$-Link Path problem and prove that our weight function has the {\it concave Monge property} (see Lemma~\ref{lem:Concave_Monge_proof}).
First, we construct an edge-weight oracle in $\Oh(nm\log(n))$ pre-processing time that gives a response in $\Oh(\log m)$ time.
Then, for $k = \Omega(\log(m))$, using the algorithm of Schieber~\cite{Schieber98} we find a solution in $m \cdot 2^{\Oh(\sqrt{\log(k)\log\log(m)})} = m^{1+o(1)}$ time and oracle accesses.
For the case $k = \Oh(\log(m))$ we use the $\Oh(mk)$ time algorithm of Aggarwal and Park~\cite{AggarwalP88} to achieve $\Oh(m\log(m))$ time and oracle accesses.
Overall, the running time is $\Oh(nm\log(n)+m^{1+o(1)} \log (m)+m\log^2(m))$, which is almost linear in the input size.

We believe our algorithm will be efficient in practice as the $nm\log(n)$ part comes from sorting (hence, it is very efficient in practice) and the $m^{1+o(1)}$ part is a lower order term when $n = \Omega(m)$ (which is typical in many voting scenarios).
All other components have running time $\Oh(nm)$ and are efficient.
Moreover, one can also replace the $m^{1+o(1)}$ part with any other (asymptotically slower) implementation for Minimum Weight $t$-Link Path (such as the ones in~\cite{AggarwalST94}) to potentially obtain better practical performance for specific applications.

We then investigate the computational complexity of CC when the preference profile is not SP but close to SP in terms of the candidate-deletion operation.
We assume a subset of $d$ candidates is given, whose removal makes an instance SP.
We call it the {\it candidate-deletion set}.
We use this general problem formulation, as for some popular utility functions a smallest candidate-deletion set is easy to compute, but for some it is $\classNP$-hard.
One of the standard motivations for considering nearly structured preferences is preference elicitation with small errors possible~\cite[Chapter 10.5]{ElkindLP17trends}.
Here, we present one more scenario when nearly SP preferences may appear:
when a few new politicians come to the political scene and the voters are not sure how to place them on the left-right political spectrum.
Clearly, the new politicians form a candidate-deletion set (assuming that the previously known politicians are placed on an SP-axis due to their longer public activity).

It is known that given the candidate-deletion set, CC can be solved in $\Ohstar(2^d)$ time \cite{MisraSV17}, and thus is FPT with respect to $d$.
Importantly, we show that the base of the power is optimal assuming the Strong Exponential Time Hypothesis.
Notice that for $d = \Oh(\log(nm))$ our algorithm runs in polynomial-time.
On the other hand, we show in Theorem~\ref{thm:cc-nsp-hardness} that if $d$ is slightly larger (e.g., $\omega(\log n)$ or $\omega(\log m)$), then polynomial-time solvability of CC contradicts the Exponential Time Hypothesis (ETH).
Overall, we derive polynomial-time algorithm for CC for almost all values of $d$ for which such a polynomial-time algorithm may exist under ETH.

We adapt the above algorithm to Thiele rules by extending the Integer Linear Programming approach of Peters~\cite{Peters18} to generalized Thiele rules (allowing different weight sequences for each voter) on SP preference profiles.
This allows us to pre-elect some winning candidates by guessing the winners from the candidate-deletion set.
As a result we obtain an $\Ohstar(2^d)$-time algorithm.
Unlike the case for CC, we were not able to prove that the base of the power cannot be improved.
However, we show that there is no $\Ohstar(2^{o(d)})$-time algorithm under ETH for each non-constant $w$-Thiele rule.
Using this, we show that if $d$ is allowed to be slightly larger than $\Oh(\log(nm))$ then polynomial-time solvability of any non-constant $w$-Thiele rule contradicts ETH.

\subsection{Related Work}\label{subsec:related-work}

First, we present two papers that are the most relevant from a technical point of view.

\paragraph{The paper of Constantinescu and Elkind~\cite{ConstantinescuE21}.}
Recently, Constantinescu and Elkind~\cite{ConstantinescuE21} also used fast algorithms for Minimum Weight $t$-Link Path to solve CC, but on SC preference profiles.
In their reduction, they construct an instance of Minimum Weight $t$-Link Path on an $\Oh(n)$-node graph with the concave Monge property while the graph we construct has $\Oh(m)$ nodes.
The best algorithm for solving Minimum Weight $t$-Link Path on an $\Oh(n)$-node graph with the concave Monge property runs in $n^{1+o(1)}$ time~\cite{Schieber98}, but Constantinescu and Elkind used an $\Oh(m)$ time algorithm to query the weight of each edge, so they end up getting an $mn^{1+o(1)}$ running time.
In our algorithm, we first pre-process the preference profile in $\Oh(nm \log(n))$ time, and then an algorithm is able to query the weight of each edge in $\Oh(\log(m))$ time, so we get an $\Oh(nm \log(n) + m^{1+o(1)})$ overall running time.
Even though both algorithms have near-linear running times, our algorithm is arguably faster since the $n^{o(1)}$ factor of their algorithm (and the $m^{o(1)}$ factor of our algorithm) could be super poly-logarithmic.
Therefore, in the natural case where a preference profile is both SC and SP~\cite{ElkindFS20}, it is potentially more efficient to use our algorithm.
The difference enlarges if an SP-axis and an SC-axis are not given.
Indeed, for ballots given as linear orders\footnote{
It is the case for, e.g., utility functions defined as {\it committee scoring rules}~\cite{ElkindFSS17} such as {\it Borda scores}.
}
the problem of finding an SC-axis is known to be computable in $\Oh(nm^2)$ time~\cite{BredereckCW13} while finding an SP-axis can be done in $\Oh(nm)$ time (even if only one vote is a linear order)~\cite{FitzsimmonsL20}.

\paragraph{The paper of Misra, Sonar, and Vaidyanathan~\cite{MisraSV17}.}
Misra et al.~\cite{MisraSV17} presented an $\Ohstar(2^d)$ time algorithm for CC-SP, whose main idea is to try all subsets of a given candidate-deletion set as winners and compute remaining winners.
This construction is well-suited for generalization to Thiele rules, so we provide a proof similar to theirs for the CC-SP case for completeness.
We note that Misra et al.\ also considered other research directions in their paper, e.g., voter-deletion distance to SP and egalitarian version of CC.\\

Next, we present related works that are less technically related to our paper but give a broader picture on the topic.

\paragraph{Other nearly SP measures.} Different notions of nearly single-peaked preference profiles have been proposed~\cite{ElkindLP17trends,ErdelyiLP17}.
Let us mention just a few of them.
{\it Voter-deletion distance} is the minimum number of voters one has to delete from the instance to achieve the SP property.
{\it Swap distance} is the minimum number of swaps one has to make within the votes (in a model with ordinal ballots) to obtain the SP property.
Computing some of SP measures appears to be \classNP-hard,
in particular, it is the case for both notions mentioned above~\cite{ErdelyiLP17}.
Hence, some FPT and approximation algorithms have been proposed~\cite{ElkindL14}.

\paragraph{Generalizations of SP to other graphs.}
There are generalizations of SP to more complex graphs, e.g., a circle or a tree.
For example, all Thiele rules are polynomial-time solvable for SP preference profiles on a line, and also on a circle~\cite{PetersL20}.
Peters et al.~\cite{PetersYCE22} generalized the algorithm of Betzler et al.~\cite{BetzlerSU13} to SP on a tree, but their running time is $\poly(n,m^\lambda,k^\lambda)$, where $\lambda$ is the number of leaves.
Actually they showed that CC-SP is $\classNP$-hard on a tree, even for Borda scores, and they ask if the problem is FPT w.r.t.~$\lambda$ (as their algorithm is only XP w.r.t.~$\lambda$).

\paragraph{Clustering problems.}
CC is equivalent to a maximization version of classic clustering problem called \kmedian (in discrete but not necessarily metric space)~\cite{ByrkaSS18}.
Therefore, CC may be seen as a clustering method in which we partition voters into $k$ clusters, where cluster centers are defined by the winning committee members and the voters from a common cluster have the same most preferred committee member.
Procaccia et al.~\cite{ProcacciaRZ08} showed \classNP-hardness for the Monroe rule~\cite{Monroe95} which, in addition to CC, requires to have a balanced representation assignment.
The Monroe rule is related to the \ProblemName{Capacitated $k$-Median} clustering problem (see e.g.\ the paper of Cohen-Addad and Li~\cite{CohenL19}) with all cluster center capacities equal to $\lceil \frac{n}{k}\rceil$ but in the case of the Monroe rule there are also lower-bounds on the capacities used: $\lfloor \frac{n}{k}\rfloor$.

\section{Preliminaries}

First we introduce some notations and define the computational problems formally.
The $\tO$ notation suppresses factors subpolynomial in the input size, i.e., factors $(nm)^{o(1)}$.
The $\Ohstar$ notation suppresses factors polynomial in the input size, i.e., factors $(nm)^{\Oh(1)}$.

\paragraph{Elections, misrepresentation, approval ballots.}
An {\it election} is a pair $(V,C)$ consisting of a set $V$ of $n$ {\it voters} and a set $C$ of $m$ {\it candidates}.

A {\it misrepresentation function} $r: V \times C \rightarrow \reals_{\geq 0}$ measures how much a voter is misrepresented by a particular candidate ($\reals_{\geq 0}$ denotes a set of non-negative real numbers)\footnote{We consider the real-RAM model of computation in this paper.
Our algorithms can run perfectly on a typical word-RAM machine if the misrepresentations are given as integers. }.
Its dual measure is the {\it utility function} $u: V \times C \rightarrow \reals_{\geq 0}$ that indicates how well a voter is represented by a particular candidate.
A tuple $(V,C,r)$ or $(V,C,u)$ is called a {\it preference profile} or just a {\it profile} or {\it preferences}.

Balloting a misrepresentation function fully is costly for a voter (needs to provide $m$ exact numbers).
A simpler and one of the most popular type of a misrepresentation function is based on {\it approval ballots}.
In {\it approval voting} a voter $v$ gives an approval ballot as a vote, i.e., a subset $A_v$ of candidates that $v$ approves.
In the case of approval voting we obtain an {\it approval utility function}, i.e., $u(v,c) = 1$ iff $c \in A_v$ and $u(v,c) = 0$ iff $c \notin A_v$.
Furthermore, we define an {\it approval misrepresentation function} as $r(v,c) = 1-u(v,c)$.

Another popular misrepresentation function is based on {\it Borda scores}.
We call $r$ a {\it Borda misrepresentation function} if for each voter $v$ we have that $r(v,\cdot)$ is a bijection from $C$ to $\{0,1,\dots,m-1\}$.
It represents a linear ordering of candidates and counts how many candidates beat a particular candidate within the vote.

For a voter $v$, if $r(v,c)$ has different values for distinct candidates, then a vote of $v$ can be seen as a {\it linear order} over candidates.
If there exists two distinct candidates $c,c'$ such that $r(v,c) = r(v,c')$ then we call a vote of $v$ a {\it weak order}.

\paragraph{Chamberlin-Courant voting rule.}
Originally, Chamberlin and Courant~\cite{ChamberlinCourant83} defined a voting rule on the Borda misrepresentation function (currently, it is often called Borda-CC).
In this paper, analogously to Betzler et al.~\cite{BetzlerSU13}, we study a more general computational problem.
In the \chamberlincourant problem (CC) we are given an election $(V,C)$, a misrepresentation function $r$, a misrepresentation bound $R \in \reals_{\geq 0}$ and a positive integer $k$.
Our task is to find a size-$k$ subset $W \subseteq C$ (called a {\it winning committee}) such that $\sum_{v \in V} \min_{c \in W} r(v,c) \leq R$ or return NO if such subset does not exist.
Additionally, for $C' \subseteq C$, we overload the notation by writing $r(v,C') = \min_{c \in C'} r(v,c)$ and we define {\it total misrepresentation of $C'$} by $r(C') = \sum_{v \in V} r(v,C')$.
In a minimization version of the CC problem we want to find a size-$k$ committee $W$ with the minimum total misrepresentation value $r(W)$.

\paragraph{$w$-Thiele voting rules.}
We define a computational problem $w$-\thiele parameterized by an infinite non-increasing non-negative sequence\footnote{If not stated otherwise, we consider only such sequences throughout this paper and call them {\it Thiele sequences}.} $w = (w_1,w_2,\dots)$, with $w_1 = 1$, as follows.
We are given an election $(V,C)$, approval ballots $A_v$ for each voter $v \in V$, an utility bound $U \in \reals_{\geq 0}$, and a positive integer $k$.
Our task is to find a size-$k$ subset $W \subseteq C$ such that $\sum_{v \in V} \sum_{i=1}^{|A_v \cap W|} w_i \geq U$ or return NO if no such subset exists.
Additionally, for $C' \subseteq C$, we overload the notation by writing $u(v,C') = \sum_{i=1}^{|A_v \cap C'|} w_i$ and we define {\it total utility of $C'$} by $u(C') = \sum_{v \in V} u(v,C')$.
In a maximization version of $w$-\thiele we want to find a size-$k$ committee $W$ with the maximum total utility value.
We notice that, a dual minimization version of $w$-\thiele is equivalent to the {\it OWA $k$-Median} problem with $0/1$ connection costs~\cite{ByrkaSS18}.

\paragraph{Single-peaked preference profiles.}
Following Proposition 3 in~\cite{BetzlerSU13} we say that a preference profile $(V,C,r)$ is {\it single-peaked} (SP)\footnote{
Sometimes called {\it possibly single-peaked}~\cite{ElkindLP17trends} as in this paper we allow ties (weak orders).
}
if there exists a linear order $\prec$ over $C$ (called the {\it single-peaked-axis} or {\it SP-axis}) such that
for every triple of distinct candidates $c_i,c_j,c_k \in C$ with $c_i \prec c_j \prec c_k$ or $c_k \prec c_j \prec c_i$ we have the following implication for each voter $v \in V$:
$r(v,c_i) < r(v,c_j) \implies r(v,c_j) \leq r(v,c_k)$.
Note that an SP profile with approval ballots can be seen as intervals of candidates (when ordering the candidates w.r.t.~an SP-axis).

\paragraph{Candidate-deletion set.}
For a given preference profile $(V,C,r)$, a subset of candidates $D \subseteq C$ is called a {\it candidate-deletion set} if its removal from the instance makes the profile SP, i.e.,
$(V,C \setminus D,r_{-D})$ is an SP profile, where its misrepresentation function $r_{-D}: V \times (C \setminus D) \rightarrow \reals_{\geq 0}$ is defined as $r$ restricted to $V \times (C \setminus D)$.
Typically we use $d = |D|$.

\paragraph{Approval-CC-SP.}
As we mentioned earlier, Approval-CC is equivalent to $(1,0,\dots)$-\thiele.
If an SP-axis is given then the input of Approval-CC-SP can be given as two numbers per voter: both endpoints of the approval interval.
Then the input size is $\Theta(n)$.
It means that the running time of our algorithm from Theorem~\ref{thm:cc-sp}, which is $\tO(nm) = \tO(n^2)$, can be much larger than the input size.
We would like to observe that Approval-CC-SP,
i.e.\ CC with a preference restriction called {\it Voter Interval}~\cite{ElkindLP17trends},
can also be solved in almost optimal time because it is actually equivalent to finding the maximum $k$-cliques of an {\it interval graph}.
\begin{observation}
  (Application III in~\cite{Schieber98})
  Approval-CC-SP can be solved in time $n^{1+o(1)}$.
\end{observation}

\paragraph{Parameterized complexity, ETH and SETH.}
We assume basic knowledge of parameterized complexity as this is common in computational social
choice research~\cite{BredereckCFGNW14,DornS17trends}.
For a textbook on the topic see, e.g.,~\cite{CyganFKLMPPS15}.
To provide lower bounds we will use two popular conjectures on solving the satisfiability of propositional formulas in conjunctive normal form (\cnfsat):
Exponential Time Hypothesis (ETH) and its stronger version---Strong Exponential Time Hypothesis (SETH).
For formal statements see, e.g., Conjectures 14.1 and 14.2 in~\cite{CyganFKLMPPS15}, but here we state simple consequences of both conjectures as follows.
1) ETH implies that there is no $2^{o(N)}$-time algorithm for \tsat with $N$ variables and $\Oh(N)$ clauses,
2) SETH implies that there is no $(2-\eps)^{N}$-time algorithm for \cnfsat, with $N$ variables and $\Oh(N)$ clauses.

\newpage
\section{Single-Peaked Preferences and Chamberlin-Courant Voting Rule}\label{sec:cc-sp}

In this section we present our algorithm for CC-SP that is almost optimal under the assumption that an SP-axis is given in the input or at least one vote is a linear order.
In the latter case we can find an SP-axis in $\Oh(nm)$ time, otherwise, if all votes are weak orders we do this in $\Oh(nm^2)$ time~\cite{FitzsimmonsL20}.
This is a bottle-neck for our algorithm.
In such a case, we obtain the same $\Oh(nm^2)$ running time as
the original dynamic programming approach of
Betzler et al.~\cite{BetzlerSU13}.

We first label the candidates so that $c_1 \prec c_2 \prec \cdots \prec c_m$ with respect to the SP-axis.
For simplicity of our analysis, we add an artificial candidate $c_0 \prec c_1$ so that $r(v, c_0)= U$ for any $v \in V$, where $U$ is the maximum value of misrepresentation.
Similarly, we add another artificial candidate $c_{m+1} \succ c_m$ so that $r(v, c_{m+1})= U$ for any $v \in V$.
After the additions, the profile is still single-peaked and the solution to the SP-CC instance does not change.

We first reduce CC-SP to the well-studied Minimum Weight $t$-Link Path problem which is defined as follows.

\begin{definition}[Minimum Weight $t$-Link Path]
Given an edge-weighted directed acyclic graph (DAG), a source node and a target node, compute a min-weight path from the source to the target that uses exactly $t$ edges.
\end{definition}

We create a graph on vertex set $\{0, \ldots, m+1\}$.
For every $i, j$ where $0 \le i < j \le m+1$, we add an edge from $i$ to $j$ with weight $w(i, j) = r(\{c_i, c_j\})) - r(\{c_i\})$.
We let the source vertex be vertex $0$ and let the target vertex be vertex $m+1$.
We set $t = k + 1$.
To show the correctness of this reduction, we use the following claim which is a key consequence of SP.

\begin{claim}
\label{cl:CC-SP}
For any $v \in V$, $0 \le i < j \le m + 1$, and for any $C' \subseteq \{c_0, \ldots, c_{i-1}\}$, it holds $r(v, \{c_i, c_j\}) - r(v, \{c_i\}) = r(v, C' \cup \{c_i, c_j\}) - r(v, C' \cup \{c_i\})$.
\end{claim}
\begin{proof}
Suppose $c_i$ has the minimum misrepresentation of $v$ among $C' \cup \{c_i\}$, then clearly $r(v, C' \cup \{c_i\}) = r(v, \{c_i\})$ and $r(v, C' \cup \{c_i, c_j\}) = r(v, \{c_i, c_j\})$ so the equality follows.

Otherwise, there exists $c' \in C'$ such that $r(v, c') < r(v, c_i)$.
Since the profile is single-peaked and $c' \prec c_i \prec c_j$, we must have $r(v, c_i) \le r(v, c_j)$.
Therefore, $r(v, \{c_i\}) = r(v, \{c_i, c_j\})$ and $r(v, C' \cup \{c_i\}) = r(v, C' \cup \{c_i, c_j\})$, so the equality follows.
\end{proof}

The following lemma shows the correctness of our reduction from CC-SP to Minimum Weight $t$-Linked Path.

\begin{lemma}
\label{lem:reduction_correct}
Given the min-weight path from vertex $0$ to vertex $m+1$ that uses $k+1$ edges, we can compute a winning committee $W \subseteq C$ of size $k$ in linear time.
\end{lemma}
\begin{proof}
Let $P$ be any path from vertex $0$ to vertex $m+1$ that uses $k+1$ edges, and let the vertices on $P$ be $a_0, a_1, \ldots, a_k, a_{k+1}$ for some $0 = a_0 < a_1 < \ldots< a_k < a_{k+1} =  m+1$.
The weight of path $P$ is thus
\def\longestwidth{\widthof{$\stackrel{\text{Claim~\ref{cl:CC-SP}}}{=}$}} 
\begin{align*}
    \sum_{i=0}^k w(a_i, a_{i+1})
    &\stackrel{\text{def.}}{\mathmakebox[\longestwidth]{=}} \sum_{v \in V} \sum_{i=0}^k r(v, \{c_{a_i}, c_{a_{i+1}}\}) - r(v, \{c_{a_i}\})\\
    &\stackrel{\text{Claim~\ref{cl:CC-SP}}}{=} \sum_{v \in V} \sum_{i=0}^k r(v, \{c_{a_j}\}_{0 \le j \le i+1}) - r(v, \{c_{a_j}\}_{0 \le j \le i})\\
    &\stackrel{}{\mathmakebox[\longestwidth]{=}} \sum_{v \in V} r(v, \{c_{a_j}\}_{0 \le j \le m+1}) - r(v, \{c_{a_j}\}_{0 \le j \le 0})\\
    &\stackrel{}{\mathmakebox[\longestwidth]{=}} -Un + \sum_{v \in V} r(v, \{c_{a_1}, \ldots, c_{a_k}\}),
\end{align*}
which is $-Un$ plus the total misrepresentation of the committee $\{c_{a_1}, \ldots, c_{a_k}\} \subseteq C$.

Conversely, for any subset of $C$ of size $k$ with total misrepresentation $R$, we can find a path from vertex $0$ to vertex $m+1$ using $k+1$ edges with weight $-Un + R$ in the graph by following the equations reversely.
\end{proof}

By Lemma~\ref{lem:reduction_correct}, it suffices to compute the minimum weight $t$-link path of the auxiliary graph.
To do so efficiently, we use the fact that the edge weights have a certain structure.

\begin{lemma}
\label{lem:Concave_Monge_proof}
The edge weights of the Minimum Weight $t$-Link Path instance satisfy the concave Monge property, i.e., for all $1 \le i + 1 < j \le m$, it holds that $w(i, j) + w(i + 1, j + 1) \le w(i, j + 1) + w(i + 1, j)$.
\end{lemma}
\begin{proof}
By expanding the definition of $w(\cdot, \cdot)$ and focusing on the terms involving some $v \in V$, it suffices to show
\begin{align*}
        & r(v, \{c_i, c_j\}) - r(v, \{c_i\})
         + r(v, \{c_{i+1}, c_{j+1}\}) - r(v, \{c_{i+1}\})\\
    \le\hspace{4pt} & r(v, \{c_i, c_{j+1}\}) - r(v, \{c_i\})
         + r(v, \{c_{i+1}, c_{j}\}) - r(v, \{c_{i+1}\}),
\end{align*}
which simplifies to
\begin{align}
& r(v, \{c_i, c_j\}) - r(v, \{c_{i+1}, c_{j}\}) \nonumber\\
\le\hspace{4pt} & r(v, \{c_i, c_{j+1}\}) - r(v, \{c_{i+1}, c_{j+1}\}).\label{eq:Concave_Monge_proof}
\end{align}
Since the profile is single-peaked, one of the following must happen.
\begin{enumerate}
    \item $r(v, c_i) \ge r(v, c_{i+1}) \ge r(v, c_j)$.
    In this case, Equation~(\ref{eq:Concave_Monge_proof}) further simplifies to
    $$r(v, c_j) - r(v, c_j) \le r(v, \{c_i, c_{j+1}\}) - r(v, \{c_{i+1}, c_{j+1}\}).$$
    The right hand side is non-negative because $r(v, c_i) \ge r(v, c_{i+1})$, so the inequality holds.

    \item $r(v, c_i) \ge r(v, c_{i+1})$ and $r(v, c_j) \le r(v, c_{j+1})$.
    In this case, we consider a function $f(x) = \min(r(v, c_i), x) - \min(r(v, c_{i+1}), x)$.
    Since $r(v, c_i) \ge r(v, c_{i+1})$, it is not hard to verify that $f(x)$ is non-decreasing.
    Since the left hand side of Inequality~(\ref{eq:Concave_Monge_proof}) is $f(r(v, c_j))$ and the right hand side is $f(r(v, c_{j+1}))$, the inequality holds because $r(v, c_j) \le r(v, c_{j+1})$ and $f(x)$ is non-decreasing.

    \item $r(v, c_{i+1}) \le r(v, c_j) \le r(v, c_{j+1})$.
    In this case, Equation~(\ref{eq:Concave_Monge_proof}) further simplifies to $$r(v, \{c_i, c_j\}) - r(v, c_{i+1}) \le r(v, \{c_i, c_{j+1}\}) - r(v, c_{i+1}),$$
    which is true because $r(v, c_j) \le r(v, c_{j+1})$ implies $r(v, \{c_i, c_j\}) \le r(v, \{c_i, c_{j+1}\})$.
\end{enumerate}
\end{proof}

By Lemma~\ref{lem:Concave_Monge_proof}, we can use the $m^{1+o(1)}$ time algorithm for the Minimum Weight $t$-Link Path instance (we use the algorithm by Schieber~\cite{Schieber98} for the case where $k = \Omega(\log m)$ and use the algorithm by Aggarwal and Park~\cite{AggarwalP88} for the case where $k = \Oh(\log(m))$).
Note that by a brute-force algorithm, we can compute each edge weight in $\Oh(n)$ time, so the overall running time becomes $nm^{1+o(1)}$.
The subpolynomial factor in this running time is $2^{\Oh(\sqrt{\log(k) \log \log (m)})}$ when $k = \Omega(\log(m))$, which can be quite large.
The following lemma gives a way to compute the edge weights faster, and consequently improves the subpolynomial factor.

\begin{lemma}
\label{lem:data_structure}
After an $\Oh(m \sort(n, nm))$ time pre-process\-ing where $\sort(n, nm)$ is the time for sorting $n$ non-negative integers bounded by $\Oh(nm)$, we can create a data structure that supports edge weight queries in $\Oh(\log(m))$ time per query.
\end{lemma}
\begin{proof}
Recall that the edge weight of an edge from $i$ to $j$ is $r(\{c_i, c_j\})-r(\{c_i\})$.
We can easily compute $r(\{c_i\})$ for every $i$ in $\Oh(nm)$ time and store the result, so it remains to consider $r(\{c_i, c_j\})$, which was defined as $$\sum_{v \in V} \min(r(v, c_i), r(v, c_j)).$$

Fix $v$, we let $c_{l_v}$ be any of the candidates where $r(v, c_{l_v})$ is minimized.
Then we define $f(i, j)$ as
$$\sum_{\substack{v \in V \\ r(v, c_i) > r(v, c_j) \text{ OR } \\
(r(v, c_i) = r(v, c_j) \text{ AND } i \le l_v)}} \hspace{-25pt} r(v, c_j).$$
We define $g(i, j)$ as
$$\sum_{\substack{v \in V \\ r(v, c_i) < r(v, c_j) \text{ OR } \\
(r(v, c_i) = r(v, c_j) \text{ AND } i > l_v)}} \hspace{-25pt} r(v, c_i)$$
to cover the remaining cases.
Clearly, $w(i, j) = f(i, j) + g(i, j)$, so we can focus on $f(i, j)$ from now on as we can handle $g(i, j)$ similarly.

We need to have the following claim.

\begin{claim}
For any fixed $j$ and $v$, the set of $i < j$ that satisfies $r(v, c_i) > r(v, c_j) \text{\ OR\ }
(r(v, c_i) = r(v, c_j) \text{\ AND\ } i \le l_v)$ always has the form $\{0, 1, \ldots, p_{v, j}\}$ for some $p_{v, j}$.
Furthermore, we can compute the values of $p_{v, j}$ for all $v$ and $j$ in $\Oh(nm)$ time in total.
\end{claim}
\begin{proof}
Let $S$ be the set of $i < j$ such that $r(v, c_i) > r(v, c_j) \text{ OR }
(r(v, c_i) = r(v, c_j) \text{ AND } i \le l_v)$.
We aim to show that if $x \in S$, then any integer $y$, where $0 \le y < x$, is also in $S$.

Then there are two cases:
\begin{itemize}
    \item In the first case, $r(v, c_x) > r(v, c_j)$.
    In this case, we must have $r(v, c_y) \ge r(v, c_x)$ because the profile is single-peaked.
    Therefore, $r(v, c_y) > r(v, c_j)$, and thus $p \in S$.
    \item In the second case, $r(v, c_x) = r(v, c_j)$.
    In this case, clearly $y \le l_v$, so it suffices to show $r(v, c_y) \ge r(v, c_j) = r(v, c_{x})$.
    For the sake of contradiction, suppose $r(v, c_y) < r(v, c_x)$.
    Since $y < x \le l_v$, $r(v, c_y) < r(v, c_x)$ implies $r(v, c_x) \le r(v, c_{l_v})$.
    Also, since $r(v, c_{l_v})$ is minimum, we must have $r(v, c_x) = r(v, c_{l_v})$, i.e., $r(v, c_x)$ is also a minimum, and it contradicts to $r(v, c_y) < r(v, c_x)$.
\end{itemize}

For each $v$, if we iterate $j$ from $0$ to $m+1$, the value of $p_{v, j}$ will first (weakly) increase and then (weakly) decrease, due to single-peakedness.
More formally, we aim to show that for any $j_1 < j_2 < j_3$, $p_{v, j_1} > p_{v, j_2} \Rightarrow p_{v, j_2} \ge p_{v, j_3}$.
Assume for the sake of contradiction that $p_{v, j_1} > p_{v, j_2}$ and $p_{v, j_2} < p_{v, j_3}$ are both true.
In other words, there exists $i$ such that $i \le p_{v, j_1}$ and $i \le p_{v, j_3}$ while $i > p_{v, j_2}$.
Since $i > p_{v, j_2}$, one of the following must happen:
\begin{itemize}
    \item $i \ge j_2$.
    This is not possible since $i \le p_{v, j_1}$ implies $i < j_1 < j_2$.
    \item $i < j_2$ and $r(v, c_i) < r(v, c_{j_2})$.
    Since $i < j_2 < j_3$, the single-peakedness of the profile implies that $r(v, c_{j_2}) \le r(v, c_{j_3})$.
    Therefore, $r(v, c_i) < r(v, c_{j_3})$, which contradicts to $i \le p_{v, j_3}$.
    \item $i < j_2$, $r(v, c_i) = r(v, c_{j_2})$ and $i > c_{l_v}$.
    In this case, if $r(v, c_i)$ equals $r(v, c_{j_1})$, $i$ must be less than or equal to $c_{l_v}$, thus leading to a contradiction.
    Therefore, $r(v, c_i) > r(v, c_{j_1})$.
    Thus, $r(v, c_{j_2}) > r(v, c_{j_1})$.
    By single-peakedness, $r(v, c_{j_3}) \ge r(v, c_{j_2})$.
    Thus, $r(v, c_{j_3}) \ge r(v, c_i)$.
    Since $i \le p_{v, j_3}$, it is not possible that $r(v, c_{j_3}) > r(v, c_i)$.
    Therefore, $r(v, c_{j_3}) = r(v, c_i)$ and $i \le c_{l_v}$, leading to a contradiction.
\end{itemize}

Therefore, by increasing or decreasing the value of $p_{v, j}$ accordingly after we increment $j$, we can efficiently compute $p_{v, j}$ for all $j$.
The total amount of changes to the values of $p_{v, j}$ over all values of $j$ for a fixed $v$ is $\Oh(m)$.
Thus, this takes $\Oh(nm)$ time overall to compute all $p_{v, j}$.
\end{proof}

Now we can use $p_{v, j}$ to reformulate $f(i, j)$: for any $i < j$, $f(i, j) = \sum_{v \in V, i \le p_{v, j}} r(v, c_j)$.
Conceptually, for every $j$, we consider $n$ numbers $p_{v, j}$, where $p_{v, j}$ has a weight $r(v, c_j)$.
Then for each $i$, we just need to report the total weights of numbers that are at least $i$.
To achieve this task, we first sort $n$ numbers $n \cdot p_{v, j} + v$ for $v \in V$.
We will be able to retrieve $v$ from each $n \cdot p_{v, j} + v$ in the sorted array, and thus retrieve $r(v, c_j)$.
Then we compute the suffix sums of $r(v, c_j)$ in the sorted order.
For every query $i$, we can easily find the smallest $p_{v, j}$ that is greater than or equal to $i$ in $ \Oh(\log(m))$ time
(for instance, we can first remove duplicated $p_{v, j}$ during pre-processing to guarantee that there are at most $m$ of them, then use binary search to find such a $p_{v, j}$), then use the suffix sum associated with it as the answer.

Therefore, the time complexity for the pre-processing is upper bounded by the cost of sorting $n$ non-negative integers bounded by $\Oh(nm)$, multiplied by a factor of $m$, and the query time is $\Oh(\log(m))$.
\end{proof}

In the real-RAM model of computation we consider in this paper, we can use any standard sorting algorithm to sort $n$ integers in $\Oh(n \log (n))$ time, so the pre-processing time of Lemma~\ref{lem:data_structure} becomes $\Oh(nm \log(n))$.
In the word-RAM model of computation, faster algorithms for sorting integers are known.
For instance, it is known that $\sort(n, nm) = \Oh(n \sqrt{\log \log (n)})$~\cite{HanT02}.
Also, in a common case where $m = \poly(n)$, $\sort(n, nm) = \Oh(n)$ by radix sort.
One may use such faster integer sorting algorithms to obtain faster running time for the overall algorithm in the word-RAM model of computation, when all misrepresentations are given as integers (e.g.\ when we use Borda misrepresentation function).

Combining Lemma~\ref{lem:data_structure} with previous ideas we obtain the following theorem.

\begin{theorem}\label{thm:cc-sp}
  We can solve \chamberlincourant on single-peaked preference profiles in time $\Oh(nm\log (n) + m^{1+o(1)})$ assuming that an SP-axis is given or at least one vote is a linear order.
\end{theorem}

\section{Nearly Single-Peaked Preferences and Chamberlin-Courant Rule}\label{sec:nsp-cc}

Betzler et al.~\cite{BetzlerSU13} showed that Approval-CC is $\classW[2]$-hard w.r.t.~$k$.
Their reduction is from the \hittingset problem (HS), in which we are given a family $\calF = \{F_1, \dots F_n\}$ of sets over a universe $U = \{u_1, \dots, u_m\}$ and a positive integer $k$, and the task is to decide whether there exists a size-$k$ subset $U' \subseteq U$ such that $U' \cap F_i \neq \emptyset$ for every $F_i \in \calF$.

Actually, the reduction presented by Betzler et al.~\cite{BetzlerSU13} also works  for the Monroe rule.
This makes the reduction more complex than it is needed for Approval-CC.
A simple reduction from HS to Approval-CC has a strict correspondence between a universe and candidates, and between a family of sets and voters.
A voter approves a candidate if the corresponding set to the voter contains the corresponding element to the candidate.
The misrepresentation bound is equal to $0$ as this will require to cover (hit) all the voters (sets).
Because of this correspondence, Approval-CC with $R=0$ is equivalent to HS.

It was observed recently~\cite{GuptaJST21} that, assuming {\it Set Cover Conjecture} (SCC) there is no $\Ohstar((2-\eps)^n)$-time algorithm for Approval-CC or Borda-CC, for every $\eps > 0$.
This observation follows directly from the reductions presented by Betzler et al.~\cite{BetzlerSU13} and SCC-hardness of HS~\cite[Conjecture 14.36]{CyganFKLMPPS15}.
On the other hand, for general misrepresentation function, Betzler et al.~\cite{BetzlerSU13} showed an $\Ohstar(n^n)$-time algorithm.
An interesting open question is to close this gap for general misrepresentation functions.
Note that for the case of Borda-CC, Gupta et al.~\cite{GuptaJST21} showed an $\Ohstar(2^n)$-time algorithm.

Below, we observe the SETH-hardness of CC, due to the SETH-hardness of HS~\cite{CyganDLMNOPSW16} and the reduction of Betzler et al.~\cite{BetzlerSU13}.

\begin{observation}\label{obs:seth-cc}
  Assuming the Strong Exponential Time Hypothesis, for every $\eps > 0$, there is no $\Ohstar((2-\eps)^m)$-time algorithm for CC, even for approval ballots.
\end{observation}

It means that the $\Ohstar(2^m)$-time brute-force algorithm, which checks all size-$k$ subsets of candidates, is essentially optimal for CC in terms of parameter $m$.
On the other hand, for CC-SP we showed an almost optimal algorithm (Theorem~\ref{thm:cc-sp}).
Hence, it is natural to study computational complexity of CC on preferences that are close to SP in terms of the candidate-deletion operation.

We assume that a candidate-deletion set is given as this makes the problem formulation as general as possible.
Otherwise, finding the minimum size candidate-deletion set would be a bottleneck for our algorithm as, in general, the problem is $\classNP$-hard.
The hardness follows from the observation that the problem
with approval ballots is equivalent to the problem of deleting a minimum number of columns to transform a given binary matrix into a matrix with the {\it consecutive ones property} (Min-COS-C).
For an overview on the computational complexity of Min-COS-C see, e.g.,~\cite{DomGN10}.
An interesting hardness result for Min-COS-C described therein is that an $\alpha$-approximation algorithm for Min-COS-C derives an $\alpha/2$-approximate solution to the {\it Vertex Cover} problem.
As a consequence, it is $\classNP$-hard to approximate the smallest candidate-deletion set within a factor of $2.72$~\cite{KhotMS17} and it is Unique Games Conjecture-hard to approximate within a factor of $4-\epsilon$, for every $\epsilon > 0$~\cite{KhotR08}.
The hardness holds already for approval preferences, where each voter approves at most $2$ candidates.
On the other hand, we can compute the minimum size candidate-deletion set in polynomial-time if all votes are linear
orders~\cite{Przedmojski2016AlgorithmsAE,ErdelyiLP17}.
This holds, e.g., when the misrepresentation function is derived from {\it a committee scoring rule}~\cite{ElkindFSS17} where a classic example is Borda scoring rule used in Borda-CC.

In the following theorem we show that CC is FPT w.r.t.~the size of a given candidate-deletion set\footnote{This result was presented before by Misra et al.~\cite{MisraSV17}, however our construction is well-suited for generalization to Thiele rules presented in Section~\ref{sec:thiele-nsp}, so we provide a proof for completeness.} (denoted by $D$) and the obtained algorithm is essentially optimal assuming SETH.
The main idea of the algorithm is to:
1) guess pre-elected winners among $D$;
2) delete $D$ from the instance together with appropriate modification of the instance with respect to the pre-elected winners;
3) as the modified instance appears to be a CC-SP instance, we use our algorithm from Theorem~\ref{thm:cc-sp}.

\begin{theorem}
\label{thm:cc-nsp}
  We can solve \chamberlincourant with a given candidate-deletion set of size $d$ in time $\Ohstar(2^d)$.
  Furthermore, assuming the Strong Exponential Time Hypothesis, for any $\eps > 0$ there is no $\Ohstar((2-\eps)^d)$-time algorithm that solves the problem, even for approval ballots.
\end{theorem}
\begin{proof}
Let $I = (V,C,r,k)$ be an instance of the optimization version of CC and $D$ be a candidate-deletion set of size $d$.
For the sake of analysis let us fix an optimal solution $W^*$ to the instance $I$.

{\bf Algorithm.}
First, we would like to find $W_D^* = W^* \cap D$ which we call {\it pre-elected winners}.
It is enough to consider all subsets $W_D \subseteq D$ of size at most $k$.
We have at most $\sum_{i=0}^{\min(k,d)} \binom{d}{i} \leq 2^d$ many such subsets and $W_D^*$ is one of them.
For each subset $W_D$, we will remove $D$ from the instance and find remaining $k' = k-|W_D|$ winners that minimizes total misrepresentation of a committee consisted of the members of $W_D$ and the remaining $k'$ candidates.
We will store all committees constructed for each considered $W_D$ and we will return the committee with the smallest total misrepresentation.
Notice that in the case where $W_D = W_D^*$ we will construct $W^*$ (or another optimal solution to the instance $I$) as one of the potential solutions.
Below, we describe how to fill optimally a pre-elected committee $W_D$.

If $|W_D| = k$ then we simply save $W_D$ as one of potential solutions of our algorithm.
Otherwise, if $|W_D| < k$, we need to choose additional $k' = k-|W_D|$ winners to the committee among a subset $C' = C \setminus D$.
For this purpose, we define a new instance $I' = (V,C',r',k')$ of CC, where $r': V \times C' \rightarrow \reals_{\geq 0}$ is defined as follows:
\[
  r'(v,c) = \min(r(v,W_D), r(v,c)) = r(v,W_D \cup \{c\}).
\]
The misrepresentation function $r'$ is constructed in such a way that misrepresentation of $v$ from a candidate $c \in C'$ is not greater than misrepresentation of $v$ by any pre-elected candidate, i.e., a candidate from $W_D$.
Moreover, we have $r'(v,W') = r(v,W_D \cup W')$.
It means that an optimal solution to $I'$ has the same misrepresentation as an optimal extension of the pre-elected committee $W_D$, as desired.
Note that, when $W_D = W_D^*$, an optimal solution to $I'$ has the same total misrepresentation as an optimal solution $W^*$.
Below we show formally that the obtained preference profile $(V,C',r')$ is single-peaked.

\begin{lemma}
  The preference profile $(V,C',r')$ is single-peaked.
\end{lemma}
\begin{proof}
  By the definition of $D$, we know that $(V,C \setminus D,r_{-D})$ is SP.
  It means that there exists a linear order $\prec$ over $C \setminus D = C'$ such that $r_{-D}$ has the single-peakedness conditions.
  We will show that the linear order $\prec$ is also an SP-axis for $r'$.
  For this purpose, let us consider a triple of distinct candidates $a, b, c \in C'$ such that $a \prec b \prec c$ and let $v \in V$ be any voter.
  If $r'(v,a) < r'(v,b)$ then by the definition of $r'$ we have $r(v, W_D \cup \{a\}) < r(v, W_D \cup \{b\})$.
  Due to the strictness of the inequality and the definition of $r(v,\cdot)$ we obtain $r(v,a) < r(v,b)$.
  Single-peakedness of $(V,C',r)$ on the SP-axis $\prec$ implies $r(v,b) \leq r(v,c)$.
  From the definition of $r(v,\cdot)$ we obtain $r(v, W_D \cup \{b\}) \leq r(v, W_D \cup \{c\})$
  and, by the definition of $r'$, this is equivalent to $r'(v,b) \leq r'(v,c)$ as desired in the condition for single-peakedness for $(V,C',r')$.
  The condition when $c \prec b \prec a$ can be proven analogously.
\end{proof}

Therefore, the new instance is an CC-SP instance which can be solved in polynomial-time, e.g., using our algorithm from Theorem~\ref{thm:cc-sp}.
In such a way, we find the remaining $k'$ winners optimally.

{\bf Running time.}
The running time comes from a consideration of all subsets of $D$ of size at most $k$ (there are at most $2^d$ such subsets) and, for each such subset, run an algorithm that solves an CC-SP instance in polynomial-time.

{\bf Lower bound.}
Let $I = (V,C,r,k)$ be an instance of Approval-CC.
Let $D$ be any size-$(m-2)$ subset of $C$.
Note that $D$ is a candidate-deletion set, as a profile with two candidates is always SP.
Then, an $\Ohstar((2-\eps)^d)$-time algorithm for Approval-CC on an instance $I$ with a given candidate-deletion set $D$, for some $\epsilon > 0$, would solve $I$ in time at most $\Ohstar((2-\eps)^d) \leq \Ohstar((2-\eps)^m)$.
This is a contradiction with SETH due to Observation~\ref{obs:seth-cc}.
\end{proof}

A simple corollary from Theorem~\ref{thm:cc-nsp} is that CC is poly\-nomial-time solvable if $d$ is logarithmic in the input size.
\begin{corollary}
  \chamberlincourant with a given candidate-deletion set of size $d$, where $d \leq \Oh(\log(nm))$, is polynomial-time solvable.
\end{corollary}

One would ask whether we can have polynomial-time algorithm for larger values of $d$.
Unfortunately, in the following theorem we show that if $d$ is slightly larger than logarithm in the input size, say $d = \Oh(\log n \log\log m)$, then it is not possible under ETH.

\begin{theorem}
\label{thm:cc-nsp-hardness}
  Under the Exponential Time Hypothesis, there is no polynomial-time algorithm for \chamberlincourant with a given candidate-deletion set of size at most $f(n, m)$ for any function $f(n, m) = \omega(\log(n))$ or $f(n, m) = \omega(\log(m))$.
\end{theorem}
\begin{proof}
By the known ETH hardness of HS~\cite{ImpagliazzoPZ01} and the above-mentioned reduction from HS to Approval-CC, Approval-CC does not have $2^{o(N)}$ time algorithm when there are $\Oh(N)$ candidates and $\Oh(N)$ voters under ETH.
From such hard instances, we construct hard instances for every function $f(n, m) = \omega(\log(n))$ or $f(n, m) = \omega(\log(m))$.

For the first case, given an Approval-CC instance with $\Oh(N)$ candidates and $\Oh(N)$ voters, we aim to create an instance where the size of the candidate-deletion set $d$ is $\Oh(h(n, m) \log(n))$ for $h(n, m) = f(n, m) / \log(n)= \omega(1)$.
We then add some dummy voters, one by one, until $n \ge 2^{N / h(n, m)}$ so that for each dummy voter $v$, $r(v, c) = 0$ for every candidate $c$ ($r$ is still an approval misrepresentation function).
The minimum total misrepresentation will not change with the addition of these voters.
The candidate-deletion set could be the set of all candidates.
In the new instance $n \ge 2^{N / h(n, m)}$ and $d = m = \Oh(N)$.
Thus, $d = \Oh(h(n, m) \log(n))$ as desired.
Since the reduction clearly has running time $2^{o(N)}$, the new instance does not have $\poly(n) = \poly(2^{N / h(n, m)}) = 2^{o(N)}$ time algorithm under ETH.

The other case is analogous.
Given an Approval-CC instance with $\Oh(N)$ candidates and $\Oh(N)$ voters, we aim to create an instance where $d = \Oh(h(n, m) \log(m))$ for $h(n, m) = f(n, m) / \log(m)= \omega(1)$.
We then add some dummy candidates until $m \ge 2^{N / h(n, m)}$ so that for each dummy candidate $c$, $r(v, c) = 1$ for every voter $v$ ($r$ is still an approval misrepresentation function).
The minimum total misrepresentation will not change with the addition of these candidates.
Also, note that if we delete the original $\Oh(N)$ candidates, the profile will be single-peaked.
Thus, $d \le \Oh(N) = \Oh(h(n, m) \log m)$ as desired.
Under ETH, the new instance does not have $\poly(m) = \poly(2^{N / h(n, m)}) = 2^{o(N)}$ time algorithm.
\end{proof}

\newpage
\section{Nearly Single-Peaked Preferences and Thiele Voting Rules}\label{sec:thiele-nsp}

In this section, we first provide an ETH-based lower bound on the running times of algorithms that solve any non-constant $w$-\thiele.

The first proof of \classNP-hardness of all non-constant $w$-\thiele follows from \classNP-hardness proof for
the {\it OWA-Win\-ner} problem for a specific family of {\it OWA vectors}~\cite[Theorem 5]{SkowronFL16}.
To show ETH-hardness straightforwardly, we will instead use a reduction from the \independentset problem (IS).
In IS we are given a graph, a positive integer $k$ and we are asked whether there exists a size-$k$ subset of vertices such that no two of them are adjacent.
The idea to reduce from IS was used in \cite{AzizGGMMW15} for the above-mentioned {\it Proportional Approval Voting}.
Actually, their reduction works for any $w$-\thiele with $w_2 < 1$.
The idea was extended by Jain et al.~\cite{JainST20} to other non-constant $w$-Thiele rules when $w_2=1$ by introducing a proper number of dummy candidates.
They presented their result for a more general problem and, actually, they did not include a formal construction in their paper,
so in the proof of Theorem~\ref{thm:eth-thiele} we include a simplified proof that works specifically for all non-constant $w$-\thiele.

In our reduction, we have $m = |V(G)|+\Oh(1)$ and $n = \frac{3}{2}|V(G)|$, where $V(G)$ is the set of vertices in an IS instance on a $3$-regular graph.
Therefore, since IS on $3$-regular graphs does not have $2^{o(|V(G)|)}$ time algorithms under ETH~\cite[Theorem~5]{Amiri21}, we obtain the following theorem.

\begin{theorem}
\label{thm:eth-thiele}
  Assuming the Exponential Time Hypothesis, there is neither $\Ohstar(2^{o(m)})$ nor $\Ohstar(2^{o(n)})$-time algorithm for any non-constant $w$-\thiele.
\end{theorem}
\begin{proof}
  First, we provide a formal construction of the reduction from IS to $w$-\thiele for any non-constant Thiele sequence $w$.
  Its idea was described in~\cite{JainST20} for an even more general problem (related to a participatory budgeting problem which allows interactions between items we choose).
  For completeness we include a simplified construction of the reduction tailored for our needs.

  {\bf Reduction.} Let $(G,K)$ be an instance of IS, where $G$ is a $3$-regular graph with a vertex set $V(G)$ and an edge set $E(G)$, and $K$ is a positive integer.
  For a fixed Thiele sequence $w$ we construct an instance of $w$-\thiele as follows.
  For each vertex $v \in V(G)$ we define a vertex-candidate $c_v$.
  As $w$ is a fixed non-constant Thiele sequence, there exists a minimum constant $\ell \in \{2,3,\dots\}$ such that $w_\ell < 1$.
  We add $\ell-2$ dummy candidates $C_D = \{d_1, d_2, \dots d_{\ell-2}\}$ to the instance.
  Hence, the total number of candidates equals $m = |V(G)|+\ell-2 = |V(G)|+\Oh(1)$.
  For each edge $e \in E(G)$ we define a voter $v_e$ which approves exactly two candidates that are not dummy corresponding to the endpoints $x$ and $y$ of $e$ and all dummy candidates, i.e., $ A_{v_e} = \{x,y\} \cup C_D$.
  The number of voters equals $n = E(G) = \frac{3}{2}V(G)$.
  We define the desired size of the committee as $k = K + \ell-2$ and the utility bound as $U = (\ell-2)E(G) + 3K$.
  Clearly, it is a polynomial-time reduction.

  {\bf Correctness.} Let $V'$ be a size-$K$ independent set of $G$.
  We use $C(V')$ to denote the set of candidates corresponding to vertices from $V'$.
  We define a solution $W$ to the constructed instance of $w$-\thiele as candidates from $C(V')$ and all dummy candidates, i.e., $W = C(V') \cup C_D$.
  We will show that this is a feasible solution.
  Clearly $|W| = K + \ell-2 = k$.
  The total utility achieved by the solution $W$ consists of:
  $|E(G)|\sum_{i=1}^{\ell-2} w_i=|E(G)|(\ell-2)$ from dummy candidates and $3K$ from $C_D$ as each vertex from the corresponding size-$K$ independent set covers exactly $3$ edges and none of the edges is covered twice (hence each voter is covered at most $\ell-1$ times).
  Formally, the total utility achieved by $W$ can be computed as follows:
  
  \begin{equation*}
  \renewcommand\arraystretch{2.2}
  \begin{array}{lcl}
    u(W) &=& \sum_{e \in E(G)} u(v_e,W)
         = \sum_{e \in E(G)} \sum_{i=1}^{|A_{v_e} \cap W|} w_i\\
         &=& \sum_{\{x,y\} \in E(G)} \sum_{i=1}^{|(\{c_x,c_y\} \cup C_D) \cap (C(V') \cup C_D)|} w_i\\
         &=& \sum_{\{x,y\} \in E(G)} \sum_{i=1}^{|\{c_x,c_y\} \cap C(V')| + \ell-2} w_i\\
         &=& (\ell-2)|E(G)| + \sum_{\{x,y\} \in E(G)} \sum_{i=1}^{|\{c_x,c_y\} \cap C(V')|} \hspace{-10pt}w_{i+\ell-2}
    \end{array}
  \end{equation*}
  \begin{equation*}
  \renewcommand\arraystretch{2.2}
  \begin{array}{lcl}
         &=& (\ell-2)|E(G)| + \sum_{\{x,y\} \in E(G)} \sum_{i=1}^{|\{x,y\} \cap V'|} w_{i+\ell-2}\\
         &\stackrel{(*)}{=}& (\ell-2)|E(G)| + \sum_{\{x,y\} \in E(G)} |\{x,y\} \cap V'| \cdot w_{\ell-1}\\
         &\stackrel{(**)}{=}& (\ell-2)|E(G)| + 3K \cdot 1 = U,
    \end{array}
  \end{equation*}
  where $(*)$ follows because $V'$ is an independent set, so it is never true that $|\{x,y\} \cap V'|=2$, and $(**)$ follows from the fact that $w_{\ell-1}=1$ and the number of incident edges to $V'$ is $3K$ as $G$ is a $3$-regular graph and $|V'| = K$.
  Hence, $W$ is a solution to the constructed instance.

  In the other direction, let $W$ be a feasible solution to the constructed instance.
  First, we notice that $C_D \subseteq W$.
  Otherwise, the total utility achieved by $W$ would be at most:
  $(\ell-3)|E(G)| + 3(K+1) = (\ell-2)|E(G)| + 3K + (3-|E(G)|) < U$.
  We define a size-$K$ subset of $V(G)$ as vertices that correspond to $W \setminus C_D$ and next, we argue that the constructed subset is an independent set.
  Indeed, if it is not the case, then there exists an edge $e =(x,y) \in E(G)$ such that $c_x,c_y \in W \setminus C_D$, hence the voter $v_e$ is covered exactly $\ell$ times.
  Then, the total utility achieved by $W$ is at most:
  $(\ell-2)|E(G)| + 3K-1 + w_{\ell} = U + (w_{\ell}-1) < U$.
  It is a contradiction with the fact that $W$ is a feasible solution.

  {\bf Lower bound.} To achieve a lower bound, let us assume there is an $\Ohstar(2^{o(m)})$-time algorithm for a fixed non-constant $w$-\thiele.
  Using the reduction presented above we can solve IS on a $3$-regular graph $G$ in time:
  $\Ohstar(2^{o(m)}) = \Ohstar(2^{o(|V(G)|+\Oh(1))}) = \Ohstar(2^{o(|V(G)|)})$,
  which contradicts ETH~\cite[Theorem~5]{Amiri21}.

  Analogously we prove the other lower bound using the fact that $n = \frac{3}{2}|V(G)|$.
\end{proof}
Theorem~\ref{thm:eth-thiele} means that the $\Ohstar(2^m)$-time brute-force algorithm is essentially optimal for non-constant $w$-Thiele rules in terms of parameter $m$.
In contrast to CC, the gap for parameter $n$ is much larger in the case of Thiele rules.
Known algorithm that is FPT w.r.t.~$n$ for Thiele rules has double exponential dependence, namely $\Ohstar(2^{2^{\Oh(n)}})$ \cite{BredereckF0KN20}.
Narrowing the gap for the parameter $n$ is an interesting research direction.

It is known that $w$-\thiele is polynomial-time solvable on SP profiles by, e.g., using an {\it Integer Linear Programming} formulation (ILP) which has {\it totally unimodular matrix} (TUM)~\cite{Peters18}.
Such ILP can be solved by a polynomial-time algorithm for {\it Linear Programming}.

Analogously to CC, we study the computational complexity of $w$-\thiele where a candidate-deletion set $D$ of size $d$ is given.

In Theorem~\ref{thm:thiele-nsp} we show that $w$-\thiele is also FPT w.r.t.~$d$, with an $\Ohstar(2^d)$ running time by checking all subsets of a candidate-deletion set as pre-elected winners.
However, the way we solve an instance with pre-elected winners is completely different from what we do for CC (Theorem~\ref{thm:cc-nsp}).
In order to derive the above result we first define generalized Thiele rules in which each voter has a private Thiele sequence.
Formally, the input of \generalizedthiele is a superset of the input of regular $w$-\thiele and additionally there is also a {\it collection of~$n$ Thiele sequences}, one for each voter.
(In contrast, the Thiele sequence is not an input to $w$-\thiele.
We emphasize that $w$-\thiele is a collection of computational problems parameterized by a Thiele sequence $w$, hence we can provide results for specific Thiele sequences, e.g., in Theorem~\ref{thm:eth-thiele}).
Precisely, a collection of $n$ Thiele sequences $w$ is represented by a function of two arguments: $w: V \times \naturals^+ \rightarrow [0,1]$, where $(w(v,i))_{i \in \naturals^+}$ is a Thiele sequence for each voter $v \in V$.
For brevity we define $w_i^v = w(v,i)$.
Now, utility of a voter $v$ from a committee $C'$ is defined as $u(v,C') = \sum_{i=1}^{|A_v \cap C'|} w_i^v$.
As in $w$-\thiele, in \generalizedthiele our task is to find a size-$k$ subset $W \subseteq C$ such that the total utility $u(W)$ is at least $U$.

It is easy to see that any $w$-\thiele is a special case of \generalizedthiele.
Indeed, a given instance of $w$-\thiele is an instance of \generalizedthiele where $w$ is the Thiele sequence for each voter.

One immediate yet interesting result, is that \generalizedthiele on SP profiles is polynomial-time solvable.
The proof follows from a proper modification of the objective function of ILP presented by Peters~\cite{Peters18}.

\newpage
\begin{theorem}
\label{thm:gen-thiele-sp}
  \generalizedthiele on single-peaked preferences profiles can be solved in polynomial-time.
\end{theorem}
\begin{proof}
We use the approach presented by Peters~\cite{Peters18} for $w$-\thiele, i.e.\ the ILP formulation labeled PAV-IP therein, which has totally unimodular matrix.
We modify the objective function as follows.
Instead of
\[
  \max \sum_{v \in V} \sum_{\ell=1}^{k} w_\ell \cdot x_{v,\ell},
\]
we use
\[
  \max \sum_{v \in V} \sum_{\ell=1}^{k} w_\ell^v \cdot x_{v,\ell}.
\]
A binary variable $x_{v,\ell}$ encodes information whether $v$ has at least $\ell$ approved candidates in a solution.
Clearly, such ILP solves \generalizedthiele.
Notice that the constraints are exactly the same as in PAV-IP, therefore the constraint matrix is totally unimodular for SP profiles, so we can solve \generalizedthiele on SP profiles using any polynomial-time algorithm for Linear Programming.
\end{proof}

Now we are ready to show our main result in this section, i.e., an algorithm which is FPT w.r.t.~the size of a given candidate-deletion set that works not only for $w$-\thiele but also for \generalizedthiele.

\begin{theorem}
\label{thm:thiele-nsp}
  We can solve \generalizedthiele with a given candidate-deletion set of size $d$ in time $\Ohstar(2^d)$.
  Furthermore, assuming the Exponential Time Hypothesis, there is no $\Ohstar(2^{o(d)})$-time algorithm that solves the problem, even for any non-constant $w$-\thiele.
\end{theorem}
\begin{proof}
  Let $D$ be a given candidate-deletion set of size $d$.
  First, we guess pre-elected winners among $D$, call them $W_D$.
  Next we delete $D$ from the instance, but then the proper modification of the instance with respect to the pre-elected winners is different with that in Theorem~\ref{thm:cc-nsp}.
  Instead of modifying utility function we modify Thiele sequences for each voter depending on the number of approved candidates among pre-elected winners $W_D$.

  For each voter $v \in V$ we define a new Thiele sequence such that $\hat{w}_i^{v} = w_{i+|A_v \cap W_D|}^{v}$ for every $i$.
  Note that a given instance of \generalizedthiele with removed candidates from $D$ and modified Thiele sequences as defined above is a \generalizedthiele instance with an SP profile.
  Therefore, we can solve it in polynomial-time using Theorem~\ref{thm:gen-thiele-sp}.

  In at least one case, when trying all subsets of $D$, we consider the maximum subset of an optimal solution that is included in $D$.
  In this case we find the remaining winners making a committee at least as good as an optimal one, so also optimal.
  Because of the guessing phase, the running time is clearly $\Ohstar(2^d)$.

  {\bf Lower bound.} Analogously to the lower bound proof of Theorem~\ref{thm:cc-nsp}, for a given instance of non-constant $w$-\thiele we define $D$ as a set of all candidates except arbitrary $2$ candidates.
  After removing $D$ from the instance, we get a profile with $2$ candidates, hence an SP profile, so $D$ is a candidate-deletion set.
  Then, an $\Ohstar(2^{o(d)})$-time algorithm for \generalizedthiele on our instance with a given candidate-deletion set $D$ would solve $I$ in time at most $\Ohstar(2^{o(d)}) \leq \Ohstar(2^{o(m)}))$.
  This is a contradiction with ETH due to Theorem~\ref{thm:eth-thiele}.
\end{proof}

Analogously to the results for CC, \generalizedthiele is polynomial-time solvable if $d$ is logarithmic in the input size and is not polynomial-time solvable under ETH if $d$ is slightly larger.

\begin{corollary}
  \generalizedthiele with a given candidate-deletion set of size $d$, where $d \leq \Oh(\log(nm))$, is polynomial-time solvable.
\end{corollary}

\begin{theorem}
\label{thm:thiele-nsp-hardness}
  Under Exponential Time Hypothesis, there is no polynomial-time algorithm for \generalizedthiele with a given candidate-deletion set of size at most $f(n, m)$ for any function $f(n, m) = \omega(\log(n))$ or $f(n, m) = \omega(\log(m))$, even for any non-constant $w$-\thiele.
\end{theorem}
\begin{proof}
After we apply Theorem~\ref{thm:eth-thiele}, the remaining proof is essentially the same as the proof of Theorem~\ref{thm:cc-nsp-hardness}.
The only difference is that in the case of adding dummy voters, which approve all the candidates, we have to increase the required utility bound properly.
Indeed, each dummy voter will get exactly $\sum_{i=1}^{k} w_i$ score from any solution of size $k$.
\end{proof}

\section{Conclusions}

We showed an almost optimal algorithm for CC-SP (Theorem~\ref{thm:cc-sp}).
The bottleneck of the algorithm is finding (if not given) an SP-axis in the case where all the votes are given as weak orders.
Finding an SP-axis on weak orders takes $\Oh(nm^2)$ time~\cite{FitzsimmonsL20}---they reduce this problem to the Consecutive Ones problem and apply known algorithms for the latter problem.
Can it be improved to $\Oh(nm)$-time?

We showed that CC on general instances is FPT w.r.t.~the size of a given candidate-deletion set $D$ (Theorem~\ref{thm:cc-nsp}).
Moreover, the achieved $\Ohstar(2^{|D|})$-time algorithm is essentially optimal under SETH.
We adapted this result to Thiele rules (Theorem~\ref{thm:gen-thiele-sp}), but we do not know how much the base of the power in our $\Ohstar(2^{|D|})$-time algorithm for Thiele rules can be improved.
We know that, under ETH, we cannot hope for an $\Ohstar(2^{o(|D|)})$-time algorithm.
Can one provide a stronger lower bound (e.g., under SETH or SCC)?

OWA-based rules are generalization of both CC and Thiele rules.
$w$-OWA-Winner rule generalizes $w$-\thiele by allowing any utility values (as in CC) instead of approval ballots~\cite{SkowronFL16}.
It is tricky to define total utility that a voter gets from a committee as we have a sequence $w$ and different utilities for all the members of a committee.
It is defined by sorting these utilities and taking a dot product with $w$ (the highest utility is multiplied by the highest weight etc.)
A minimization version of the problem was studied under the name OWA $k$-Median~\cite{ByrkaSS18}.
An open question is whether such $w$-OWA-Winner is FPT w.r.t.~the size of a given candidate-deletion set.
Our approach from Theorem~\ref{thm:thiele-nsp} does not seem to work because we need to know how to modify Thiele sequences for each voter when we pre-define winners.

\section*{Acknowledgements}
We would like to thank the anonymous reviewers for their helpful comments.

Virginia {Vassilevska Williams} was supported by NSF Grants CCF-2129139 and CCF-1909429, a BSF Grant BSF:2020356, a Google Research Fellowship and a Sloan Research Fellowship.
Yinzhan Xu was partially supported by NSF Grant CCF-2129139.
Krzysztof Sornat was partially supported by
the SNSF Grant 200021\_200731/1,
the National Science Centre, Poland (NCN; grant number 2018/28/T/ST6/00366) and
the European Research Council (ERC) under the European Union’s Horizon 2020 research and innovation programme (grant agreement No 101002854).

\bibliographystyle{plain}
\bibliography{bib}

\end{document}